\documentclass[journal,twoside,web]{ieeecolor}
\usepackage{lcsys}

\usepackage{cite}
\usepackage{amsmath,amssymb,amsfonts}
\usepackage{algorithm}
\usepackage{algpseudocode}
\usepackage{graphicx}
\usepackage{textcomp}

\usepackage{mathtools}
\usepackage{cleveref}

\newtheorem{theorem}{Theorem}
\newtheorem{proposition}[theorem]{Proposition}
\newtheorem{definition}[theorem]{Definition}

\newtheorem{lemma}[theorem]{Lemma}
\newtheorem{remark}[theorem]{Remark}
\newtheorem{example}{Example}
\newtheorem{assumption}{Assumption}
\AddToHook{env/proposition/begin}{\crefalias{theorem}{proposition}}
\AddToHook{env/definition/begin}{\crefalias{theorem}{definition}}
\AddToHook{env/corollary/begin}{\crefalias{theorem}{corollary}}
\AddToHook{env/lemma/begin}{\crefalias{theorem}{lemma}}
\AddToHook{env/remark/begin}{\crefalias{theorem}{remark}}
\AddToHook{env/example/begin}{\crefalias{theorem}{example}}
\usepackage{tikz}
\usetikzlibrary{cd}
\usetikzlibrary{shapes.geometric, arrows}
\usepackage{graphicx}
\usepackage{subcaption}
	
\newcommand{\R}{\mathbb{R}}
\newcommand{\B}{\mathfrak{B}}
\newcommand{\N}{\mathbb{N}}
\newcommand{\C}{\mathbb{C}}
\newcommand{\bs}[1]{\boldsymbol{#1}}

\newcommand{\hankel}[2]{\mathcal{H}_{#2}({#1})}

\newcommand{\sigb}{\sigma^{\pm 1}}

\begin{document}
	\title{Data-driven design of deadbeat observers for higher-order LTI systems}
	\author{Aishwarya Bharti and Debasattam Pal,  \IEEEmembership{Member, IEEE}
		\thanks{Aishwarya Bharti and Debasattam Pal are with the Department of Electrical Engineering, Indian Institute of Technology Bombay, India. (e-mail: ab@ee.iitb.ac.in, debasattam@ee.iitb.ac.in).This work has been submitted to the IEEE for possible publication.}, }
	
	\maketitle
	
	\begin{abstract}
	  In this letter, we derive a necessary and sufficient condition for the existence of a deadbeat observer (DBO) for systems described by higher-order difference equations based on finite historical data. Furthermore, we develop an algorithm to construct a consistent DBO and to parametrise all DBOs directly from the historical data. Finally, we apply the proposed approach to the problem of state estimation in the presence of unknown inputs.   
	\end{abstract}
	
	
	\section{Introduction}
	\label{sec:introduction}
	\IEEEPARstart{T}{he} growing availability of data, lower storage costs, and increased computational capabilities have motivated the use of historical data for the classical state-space estimation problems (see for example  \cite{DBO_10720866,DBOadd_mishra2022data,DBOadd_wei2026distributed, DBO_fattore2025data, DBO_shi2022data,DBO_turan2021data, DBO_disaro2024delayed}). However, to the best of our knowledge, these works are limited to first-order system representation and focus on specific estimation problems, such as unknown input observer\cite{DBO_10720866}. This causes the observer to have a particular structure (see, e.g., \cite[Definition 1]{DBO_10720866}), thereby limiting the scope of the analysis to only a restrictive class of admissible observers. In contrast, we deal with systems described by higher-order difference equations, and solve for the estimation of a subset of system variables, referred to as the \emph{relevant variables}, from another subset, referred to as the \emph{measured variables}. As articulated in the fundamental works of Valcher and Willems \cite{valcher1999observer, valcher1999dead}, this approach of ours to consider systems described by higher-order difference equations provides a unified platform for addressing classical estimation problems, including state estimation, state estimation in the presence of disturbances, fault detection, and input reconstruction, by appropriately selecting the system variables\cite{valcher_bisiacco2006behavioral}. In this letter, we focus on the deadbeat observer (DBO) design for systems described by higher-order difference equations based solely on finite historical data of relevant and measured variables; that is, an observer that estimates the relevant variables exactly after a finite duration from the measured variables. Moreover, unlike approaches restricted to particular observer representations, we provide an algorithm that computes one consistent DBO and a complete parametrisation of all DBOs described by higher-order difference equations.     

	We begin by establishing a necessary and sufficient condition for the existence of a DBO using only finite historical data. We achieve this by developing a data-driven approach to the crucial notion of \emph{reconstructibility}, which is known to be equivalent to the existence of DBOs \cite[Proposition 4.2]{valcher1999dead}. In the process, we introduce the notion of \emph{local reconstructibility} which plays a central role in enabling the use of finite data to determine the existence of DBOs and to construct a DBO solely from finite historical data. Furthermore, we provide a parametrisation of the complete family of DBOs using the available finite data.
	To illustrate our results, we apply the proposed algorithm to the problem of state estimation in the presence of unknown inputs (disturbances) using only historical data. 
	Thus, our contributions are as follows:
	\begin{itemize}
		\item We introduce the notion of \emph{local reconstructibility} for restricted behaviours and derive a condition for its equivalence to the global reconstructibility (\Cref{sec:local_recons}).
		\item We provide a necessary and sufficient condition for testing global reconstructibility from a finite historical dataset under a suitable informativity assumption (\Cref{sec:datarecons}).
		\item We present an algorithm for constructing a consistent DBO together with a complete parametrisation of all DBOs from a finite historical dataset (\Cref{sec:algorithm}).
	\end{itemize}
	 	\textit{Notation:}
	We denote the sets of natural numbers (including $0$) and real numbers by  $\N$ and $\R$, respectively. For $k\in \N$, we use the boldface $\boldsymbol{k}$ to denote the set $\{0,1,\ldots,k\}$. For two sets $\mathbb{W}$ and $\mathbb{T}$, the symbol $\mathbb{W}^{\mathbb{T}}$ denotes the set of all maps from $\mathbb{T}$ to $\mathbb{W}$, i.e., $\mathbb{W}^\mathbb{T}\coloneqq \left\{f\mid f: \mathbb{T}\to \mathbb{W}\right\}$. We use $\R[\sigma]$ and $\R[\sigb]$ to denote the rings of polynomials and Laurent polynomials in $\sigma$, respectively, with coefficients in $\R$. We use $\R^{m\times n}[\sigma]$ and $\R^{m\times n}[\sigb]$ for the sets of $m\times n$ matrices over $\R[\sigma]$ and $\R[\sigb]$, respectively. For $v_i\in \R^{m_i}$ (the space of $m_i$-tuples of real numbers), $i=1,\ldots,k$, $\mathrm{col}(v_1,v_2,\ldots,v_k)\in \R^{\sum_{i=1}^{k}m_i}$ denotes their vertical concatenation.

	\section{Preliminaries}    
	   In this letter, we consider the set of trajectories of a linear time-invariant (LTI) system, defined as the solution set of a system of higher-order difference equations with constant real coefficients. 
	   We define the left shift operator $\sigma$ that acts on $\bs{w}\in (\R^q)^\N$ as $(\sigma^k \bs{w})(t)=\bs{w}(t+k)$ for $k\in \N$. Let $R_1(\sigma)\in \R^{g\times q_1}[\sigma]$ and $R_2(\sigma)\in \R^{g\times q_2}[\sigma]$ be operator matrices that act on vector-valued trajectories of appropriate sizes.
	   We consider the behaviour $\B$ consisting of the trajectories $\bs{w}\in (\R^q)^{\N}$ partitioned into \textit{measured} variables, denoted by $\bs{w}_1\in (\R^{q_1})^{\N}$, and \textit{relevant} variables, denoted by $\bs{w}_2\in (\R^{q_2})^{\N}$, where $q_1+q_2=q$, that satisfy $R_1(\sigma)\bs{w}_1+R_2(\sigma)\bs{w}_2=0$. That is,
	   \begin{align}\label{eqn:behaviour}
	   		\B\!=\!\left\{\bs{w}\!=\!\begin{bsmallmatrix}
	   			\bs{w}_1\\\bs{w}_2
	   		\end{bsmallmatrix}\in \!(\R^{q})^\N \!\bigm| \!R_1(\sigma)\bs{w}_1\!+\!R_2(\sigma)\bs{w}_2\!=\!0 \right\}.
	   \end{align}  
	  With $R(\sigma)\coloneqq\begin{bsmallmatrix}
	  	R_1(\sigma)&~ R_2(\sigma)
	  \end{bsmallmatrix}$, $\B=\ker R(\sigma)$ is a \emph{kernel representation} and $R(\sigma)$ is its \emph{kernel representation matrix}. Note that a behaviour usually admits multiple kernel representations, all equivalent under unimodular transformations \cite[Section III]{willems1991paradigms}. Since $\R[\sigma]$ is a principal ideal domain (PID), $\B$ admits a full row rank kernel representation. For a full row rank $R(\sigma)\in \R^{g\times q}[\sigma]$, the maximum among the degrees of all $g\times g$ minors is called the \emph{order/McMillan degree} of the system\cite[Section X]{willems1991paradigms}; this plays a crucial role in our analysis.

In this letter, we address the problem of obtaining a deadbeat observer (DBO) of the relevant variables from the measured variables using a previously collected dataset. A DBO of $\bs{w}_2$ from $\bs{w}_1$ for a system provides an estimate $\widehat{\bs{w}}_2$ that coincides with $\bs{w}_2$ after a finite interval of time. DBOs are classified as consistent or non-consistent \cite{valcher_bisiacco2006behavioral}; we focus on obtaining a consistent DBO and a parametrisation of all DBOs. 
\begin{definition}\label{def:DBO}
Let $\B$ be an LTI system described by \eqref{eqn:behaviour}. For $g'\in \N$, $P(\sigma)\in \R^{g'\times q_1}[\sigma]$, and $Q(\sigma)\in \R^{g'\times q_2}[\sigma]$, the system
		$	Q(\sigma)\widehat{\bs{w}}_2=P(\sigma)\bs{w}_1$
	is said to be a \emph{deadbeat observer (DBO)} of $\bs{w}_2$ from $\bs{w}_1$ for $\B$ if both the following statements hold.
		\begin{itemize}
			\item[(a)] For all $(\bs{w}_1,\bs{w}_2)\in \B$, there exists $\widehat{\bs{w}}_2$ such that $Q(\sigma)\widehat{\bs{w}}_2=P(\sigma)\bs{w}_1$.
			\item[(b)] There exists $N\in \N$ such that for all $(\bs{w}_1,\bs{w}_2)\in \B$ and all $(\bs{w}_1,\widehat{\bs{w}}_2)$ satisfying $Q(\sigma)\widehat{\bs{w}}_2=P(\sigma)\bs{w}_1$, we have $\bs{w}_2(t)-\widehat{\bs{w}}_2(t)=0$ for all $t\geqslant N$. 
		\end{itemize}
Further, the DBO is said to \textcolor{black}{be \emph{consistent}} if in addition every $(\bs{w}_1,\bs{w}_2)\in \B$ also satisfies $Q(\sigma){\bs{w}}_2=P(\sigma)\bs{w}_1$. 
	
\end{definition}

In the sequel, we use $\bs{w}_1$ and $\bs{w}_2$ to represent the measured and relevant variables, respectively. Thus, we refer to a DBO of $\bs{w}_2$ from $\bs{w}_1$ for $\B$ simply as a DBO for $\B$.
An important step in obtaining a DBO for $\B$ is to determine whether one exists or not in the first place. This is characterised by the property of \textit{reconstructibility} of $\bs{w}_2$ from $\bs{w}_1$ \cite[Proposition 4.2]{valcher1999dead}, defined below. In order to distinguish it from \emph{local} reconstructibility, defined later in this letter, for trajectories restricted over a finite interval, we call it \emph{global} reconstructibility.

\begin{definition}[Global reconstructibility]\label{def:recons}
	Consider an LTI system $\B$ given by \eqref{eqn:behaviour}. \textcolor{black}{For a fixed $N\in \N$}, we say that $\bs{w}_2$ is \emph{globally $N$-reconstructible} from $\bs{w}_1$ if $(\bs{w}_1,\bs{w}_2), (\bs{w}_1,\overline{\bs{w}}_2)\in \B$ implies that  $\bs{w}_2(t)-\overline{\bs{w}}_2(t)=0$, for all $t\geqslant N$. We say $\bs{w}_2$ is \emph{globally reconstructible} from $\bs{w}_1$ if \textcolor{black}{there exists $N\in \N$ such that} $\bs{w}_2$ is globally $N$-reconstructible from $\bs{w}_1$.
\end{definition}

The following result, derived easily from \cite[Theorem 3]{valcher_bisiacco2006behavioral}, gives equivalent conditions for the existence of a DBO.

\begin{proposition}\label{prop: recons_monomic_DBOform}
		Consider an LTI system $\B$ described by \eqref{eqn:behaviour}. The following are equivalent:
		\begin{itemize}
			\item[(i)]  $\bs{w}_2$ is globally reconstructible from $\bs{w}_1$.
			\item[(ii)] The matrix $R_2(\sigma)\in \R^{g\times q_2}[\sigma]$ is right monomic\footnote{A matrix $M(\sigma)\in \R^{p\times q}[\sigma]$ is \emph{right monomic} if $\operatorname{rank}M(\lambda)=q$ for all $\lambda\in \C\setminus\{0\}$, equivalently, if $\exists~ N\in \N$, $X(\sigma)\in \R^{q\times p}[\sigma]$ s.t. $X(\sigma)M(\sigma)=\sigma^NI_q.$ If $p=q$, $M(\sigma)$ is called square monomic\cite[Section I]{valcher_bisiacco2006behavioral}. }.
			\item[(iii)] There exist $N\in \N$ and $P(\sigma)\in \R^{q_2\times q_1}[\sigma]$ such that $R_{\rm DBO}(\sigma)\coloneqq \begin{bsmallmatrix}
				-P(\sigma)&\quad\sigma^N I_{q_2}
			\end{bsmallmatrix}$ provides a consistent DBO for $\B$ given by $	R_{\rm DBO}(\sigma)\mathrm{col}(
		\bs{w}_1,
	\widehat{\bs{w}}_2)=0.$
		\end{itemize}
	\end{proposition}
\begin{proof}
	$(i)\iff (ii)$: Refer to \cite[Theorem 3]{valcher_bisiacco2006behavioral}.\\
	$(ii)\implies (iii)$: Since $R_2(\sigma)$ is right monomic, $\exists~N\in \N$ and a matrix $X(\sigma)\in \R^{q_2\times g}$ such that $X(\sigma)R_2(\sigma)=\sigma^NI_{q_2}$. Define $R_{\rm DBO}(\sigma)\coloneqq X(\sigma)\begin{bsmallmatrix}
		R_1(\sigma)&\quad R_2(\sigma)
	\end{bsmallmatrix}\coloneqq \begin{bsmallmatrix}
	-P(\sigma)&\quad \sigma^NI_{q_2}
\end{bsmallmatrix}$. Then, every $\bs{w}=\mathrm{col}(\bs{w}_1,\bs{w}_2)\in \B$ satisfies $R_{\rm DBO}(\sigma)\bs{w}=0$. Also, for all $\mathrm{col}(\bs{w}_1, \widehat{\bs{w}}_2)\in \ker R_{\rm DBO}$, we have $\sigma^N \widehat{\bs{w}}_2=P(\sigma)\bs{w}_1=\sigma^N \bs{w}_2$ $\forall~t$. Therefore, $\bs{w}_2(t)-\widehat{\bs{w}}_2(t)=0~ \forall~ t\geqslant N$. 

\noindent$(iii)\!\!\!\!\!\!\implies\!\!\!\! (i)$: Since $R_{\rm DBO}(\sigma)$ provides a consistent DBO, every $\bs{w}\in \B$ satisfies $R_{\rm DBO}(\sigma)\bs{w}=0$. Thus, if $(\bs{w}_1,\bs{w}_2), (\bs{w}_1, \overline{\bs{w}}_2)\in \B$, then $\sigma^N(\bs{w}_2(t)-\overline{\bs{w}}_2(t))=0$. By \Cref{def:recons}, $\bs{w}_2$ is globally $N$-reconstructible from $\bs{w}_1$. 
\end{proof}	

It is shown in \cite[Lemma 4.3]{valcher1999dead} that every DBO for $\B$ admits an equivalent\footnote{Define the projected behaviour $\Pi_{\bs{w}_1}\B\coloneqq \left\{\bs{w}_1\mid \exists~ (\bs{w}_1,\bs{w}_2)\in \B \right\}$. Two DBOs, $\widehat{\B}_1$ and $\widehat{\B}_2$, of $\B$ are said to be \emph{equivalent} if $\{(\bs{w}_1,\widehat{\bs{w}}_2)\in \widehat{\B}_1\mid \bs{w}_1\in \Pi_{\bs{w}_1}\B)\}=\{(\bs{w}_1,\widehat{\bs{w}}_2)\in \widehat{\B}_2\mid \bs{w}_1\in \Pi_{\bs{w}_1}\B\}$\cite[Section 4]{valcher1999dead}.} DBO with square monomic $Q(\sigma)$. Hence, throughout this letter, we restrict our attention to such DBOs. 
\textcolor{black}{ It is a well-known fact that a conventional state-space deadbeat observer converges to the exact state trajectory in at most $n$ (order of the system) steps\cite[Chapter 6]{ogata1995discrete}. Interestingly, for general higher-order systems also, if the system admits DBOs, then one DBO exists with convergence step within $n$ steps; an algebraic proof of this is provided below. }
\begin{lemma}\label{lem:Nmax}
	Consider an LTI system $\B$ described by \eqref{eqn:behaviour}. Let $n$ be the order of the system. Then, $\bs{w}_2$ is globally reconstructible from $\bs{w}_1$ if and only if for some $N\in \boldsymbol{n}\coloneqq \{0,1,\ldots,n\}$, $\bs{w}_2$ is globally $N$-reconstructible from $\bs{w}_1$.
\end{lemma}

\begin{proof} The `if' part trivially follows from \Cref{def:recons}; we now prove the `only if' part.
		Without loss of generality, let $R(\sigma)\in \R^{g\times q}[\sigma]$ be full row rank.
	If $\bs{w}_2$ is globally reconstructible from $\bs{w}_1$, then  $g\geqslant q_2$, and  $\exists~U(\sigma)\in \R^{g\times g}[\sigma]$ unimodular such that
	$U(\sigma)R(\sigma)=\begin{bsmallmatrix}
		L_1(\sigma)&\Delta(\sigma)\\
		L_0(\sigma)&0
	\end{bsmallmatrix}\eqqcolon\overline{R}(\sigma)$ where $\Delta(\sigma)\in \R^{q_2\times q_2}[\sigma]$ is square monomic with $\det \Delta(\sigma)=\sigma^N$ for $N\in \N$\cite[Section IV]{valcher_bisiacco2006behavioral}. Since $U(\sigma)$ is unimodular, $\B=\ker \overline{R}(\sigma)$ (see \cite[Proposition III.3]{willems1991paradigms}). Further, as $\Delta(\sigma)$ is monomic, $\exists~ X(\sigma)\in \R^{q_2\times q_2}[\sigma]$ such that $X(\sigma)\Delta(\sigma)=\sigma^NI_{q_2}$. Thus, by \Cref{prop: recons_monomic_DBOform}, $R_{\rm DBO}(\sigma) \coloneqq \begin{bsmallmatrix}
	X(\sigma)&~ 0_{q_2\times g-q_2}
\end{bsmallmatrix}\overline{R}(\sigma)$ provides a consistent DBO. Hence, it suffices to show that $N\leqslant n$. 
	Recall that the order of the system equals the highest degree among all $(g\times g)$ minors of $\overline{R}(\sigma)$. 
	When $g=q_2$, then $\overline{R}(\sigma)=\begin{bsmallmatrix}
		L_1(\sigma)&~\Delta(\sigma)
	\end{bsmallmatrix}$ and $\det \Delta(\sigma)=\sigma^N$ happens to be one of the $(g\times g)$ minors; hence $N\leqslant n$. On the other hand, when $g>q_2$, since $\overline{R}(\sigma)$ is full row rank, $L_0(\sigma)$ contains a non-singular $(g-q_2)\times (g-q_2)$ submatrix $L_{0m}(\sigma)$. Then, one of the $(g\times g)$ minors of $\overline{R}(\sigma)$ is $\det\Delta(\sigma)\det L_{0m}(\sigma)$.
 Since $\det \Delta(\sigma)=\sigma^N$ and $\det L_{0m}(\sigma)\neq 0$, it follows that $N\leqslant \deg\left(\det \Delta(\sigma)\det L_{0m}(\sigma)\right)\leqslant n$. 
\end{proof}

\begin{remark}\label{rem:irrelevant}
Throughout this letter, we partition the system variables into measured and relevant variables. Systems with additional irrelevant variables, such as unknown input in state estimation, can be reduced to this setting via projection. In particular, let $\B_{\rm full}$ contains $\bs{w}\coloneqq (\bs{w}_1,\bs{w}_2,\bs{w}_3)$ where we want DBOs of $\bs{w}_2$ from $\bs{w}_1$. Since $\bs{w}_3$ is irrelevant variable, it suffices to analyse $\B=\{(\bs{w}_1,\bs{w}_2)\mid ~\!\exists~\bs{w}_3\in (\R^{q_3})^\N \mbox{ such that } (\bs{w}_1, \bs{w}_2, \bs{w}_3)\in \B_{\rm full}\}$. However, in the model-based setting, it is usually non-trivial to carry out such a projection; therefore, in this setting, the analysis of DBO is often done in the presence of the irrelevant variable with $\B_{\rm full}$. Interestingly, in the data-driven approach, the projection of $\B_{\rm full}$ on $(\bs{w}_1, \bs{w}_2)$ can be trivially carried out by simply ignoring the data corresponding to the irrelevant variable $\bs{w}_3$. This justifies our omission of irrelevant variables in this letter.
\end{remark}
\section{Local reconstructibility}\label{sec:local_recons}
In this letter, we are interested in ascertaining global reconstructibility of $\bs{w}_2$ from $\bs{w}_1$ and obtaining a DBO for $\B$, both based only on a finite dataset. To this end, we introduce the notion of \emph{restricted behaviour}, defined as the set of trajectories of the given behaviour truncated to a finite time interval. The restricted behaviour is a finite-dimensional vector space and, therefore, admits a finite basis \cite{willems1991paradigms}. In this section, we relate global reconstructibility to a property of the restricted behaviour, which we term \emph{local reconstructibility}. This forms the basis of the finite data-driven approach in later sections.

In order to formalise the notion of restricted behaviours, let $k\in\N$ be arbitrary.
For any trajectory $\bs{w}=(\bs{w}_1,\bs{w}_2)\in (\R^q)^{\N}$, the symbol $\bs{w}|_k$, called {\em the $k$-restricted $\bs{w}$}, denotes the map $\bs{w}|_k:\boldsymbol{k}\rightarrow\R^q$
defined as	$
\bs{w}|_k(t)=(\bs{w}_1|_k(t), \bs{w}_2|_k(t))\coloneqq(\bs{w}_1(t),\bs{w}_2(t))\mbox{ for all }0\leqslant t \leqslant k.
$
Let $\B$ be the behaviour of an LTI system. The {\em $k$-restricted behaviour $\B_k$} is defined as
$\B_k\coloneqq \left\{\bs{w}|_k:\boldsymbol{k}\rightarrow\R^q \bigm| \bs{w}\in\B\right\}.
$
 By \Cref{prop: recons_monomic_DBOform}, $\bs{w}_2$ is globally reconstructible from $\bs{w}_1$ if and only if every $\bs{w}\in \B$ satisfies $\sigma^N\bs{w}_2=P(\sigma)\bs{w}_1$ where $N\in \N$ and $P(\sigma)\coloneqq\sum_{i=0}^{d}P_i\sigma^i,~d\in \N,~ P_i\in \R^{q_2\times q_1}$. Thus, we have $\bs{w}_2(N)=\sum_{i=0}^{d}P_i\bs{w}_1(i)$. This motivates the following definition.

\begin{definition}[$k$-local $N$-reconstructibility]\label{def:loc-recons}
	Consider an LTI system $\B$ described by \eqref{eqn:behaviour}. Let $k\in \N$ and $N\in \boldsymbol{k}$. We say that $\bs{w}_2$ is $k$-locally $N$-reconstructible from $\bs{w}_1$ if $(0,\bs{w}_2|_k)\in \B_k$ implies $\bs{w}_2(N)=0$.
\end{definition}

Let $k\in \N$ be arbitrary and $N\in \bs{k}$, then notice that global reconstructibility does not always imply $k$-local $N$-reconstructibility. Indeed, consider the following example.
\begin{example}\label{eg:no_const_DBO}
	Consider the behaviour
	$			\B=\left\{\!\!	\begin{bsmallmatrix}
					\bs{x}_1\\\bs{x}_2\\\bs{x}_3
				\end{bsmallmatrix}\in (\R^3)^\N \bigm| \begin{bsmallmatrix}
					\sigma^2-1&0&\quad	\sigma(\sigma+1)\\
					0&\sigma^2+1&\quad	\sigma(\sigma-1)
				\end{bsmallmatrix}		\begin{bsmallmatrix}
					\bs{x}_1\\\bs{x}_2\\\bs{x}_3
				\end{bsmallmatrix}	=0 \!\right\}$
	with the partition $\bs{w}_1=(\bs{x}_1,\bs{x}_2)$ and $\bs{w}_2=\bs{x}_3$.
	Since $\begin{bsmallmatrix}
		\sigma(\sigma+1)\\
		\sigma(\sigma-1)
	\end{bsmallmatrix}$ is right monomic, by \Cref{prop: recons_monomic_DBOform}, $\bs{w}_2$ is globally reconstructible from $\bs{w}_1$. Indeed, a (consistent) DBO is given by
	$\begin{bsmallmatrix}
		0.5(\sigma^2-1)&~~-0.5(\sigma^2-1)&~~\sigma
	\end{bsmallmatrix}\begin{bsmallmatrix}
		\bs{w}_1\\\widehat{\bs{w}}_2
	\end{bsmallmatrix}=0$. 
However, note that $\bs{w}\in \B$ with $\bs{w}_1=
\left\{
\begin{psmallmatrix}0 \\ 0\end{psmallmatrix},
\begin{psmallmatrix}0 \\ 0\end{psmallmatrix},
\begin{psmallmatrix}-2 \\ 0\end{psmallmatrix},
\begin{psmallmatrix}-2 \\ 0\end{psmallmatrix},\ldots
\right\}$ and $\bs{w}_2=\{1,1,1,1,\ldots\}$. Thus, by \Cref{def:loc-recons}, for $k=0,1$, $\bs{w}_2$ is not $k$-locally $N$-reconstructible from $\bs{w}_1$ for any $N\in \bs{k}$.
\end{example} 

Nonetheless, the converse is always true (\Cref{lem:eventoglob}). First, we prove that $k$-local $N$-reconstructibility implies the existence of a local DBO. This plays a pivotal role in the sequel.

\begin{lemma}\label{lem:local_DBO}
	Consider an LTI system $\B$ described by \eqref{eqn:behaviour}. Let $k\in \N$ and $N\in \bs{k}$. Then, $\bs{w}_2$ is $k$-locally $N$-reconstructible from $\bs{w}_1$ if and only if there exists a matrix $R_{\rm LDBO}\in \R^{q_2\times (q_1(k+1))}$ such that every $\bs{w}=(\bs{w}_1,\bs{w}_2)\in \B$ satisfies $\bs{w}_2(N)=R_{\rm LDBO}\mathrm{col}(\bs{w}_1(0),\bs{w}_1(1),\ldots,\bs{w}_1(k))$.
\end{lemma}
 \begin{proof}
 (If) Assume every $\bs{w}\in \B$ satisfies $\bs{w}_2(N)=R_{\rm LDBO}\mathrm{col}(\bs{w}_1(0),\bs{w}_1(1),\ldots,\bs{w}_1(k))$. Thus, for all $\bs{w}\in \B$, $\bs{w}_1|_k=0$ implies $\bs{w}_2(N)=0$. By \Cref{def:loc-recons}, $\bs{w}_2$ is $k$-locally $N$-reconstructible from $\bs{w}_1$. 

 (Only if)  Note that $\B_k$ is a subspace of $\R^{q(k+1)}$ since $\bs{w}|_k=(\bs{w}_1|_k,\bs{w}_2|_k)$ can be represented as $\mathrm{col}(\bs{w}_1(0),\bs{w}_1(1),\ldots,\bs{w}_1(k),\bs{w}_2(0),\ldots,\bs{w}_2(k))\in \R^{q(k+1)}$. Let $d_k=\dim \B_k$. Then, there exist $V_1\in \R^{q_1(k+1)\times d_k}$ and $V_2\in \R^{q_2(k+1)\times d_k}$ such that $\B_k=\mathrm{colspan}\begin{bsmallmatrix}
 	V_1\\V_2
 \end{bsmallmatrix}$.
 Let $V_{2,N}\in \R^{q_2\times d_k}$ be the submatrix of $V_2$ formed by rows with row index $i$ satisfying $(q_2N+1)\leqslant i\leqslant (q_2N+q_2)$. Then, for all $\bs{w}\in \B$, $\bs{w}_2(N)\in \mathrm{Im}V_{2,N}$. 
 If $\bs{w}_2$ is  $k$-locally $N$-reconstructible from $\bs{w}_1$, then for all $\bs{w}\in \B$, $\bs{w}_1|_{k}=0$ implies $\bs{w}_2(N)=0$. Equivalently,
 $\ker V_1\subseteq \ker V_{2,N}$ and hence $\mathrm{rowspan} V_1\supseteq \mathrm{rowspan} V_{2,N}$. Thus, there exists $R_{\rm LDBO}\in \R^{q_2\times (q_1(k+1))}$ such that 
 $ V_{2,N}=R_{\rm LDBO}V_1$. Hence, the claim follows.
 \end{proof}	

By \Cref{lem:local_DBO}, the presence of the local DBO implies that if $\bs{w}_2$ is $\tau$-locally $N$-reconstructible from $\bs{w}_1$ then $\bs{w}_2$ is $k$-locally $N$-reconstructible from $\bs{w}_1$ for all $k\geqslant \tau$. This motivates us to define eventual local $N$-reconstructibility below.

\begin{definition}[Eventual local $N$-reconstructibility]\label{def:even-loc-recons}
	Let $\B$ be an LTI system and let $N\in \N$. We say $\bs{w}_2$ is eventually locally $N$-reconstructible from $\bs{w}_1$ if $\exists~\tau \in \N$ such that $N\in \bs{\tau}$ and $\bs{w}_2$ is $k$-locally $N$-reconstructible from $\bs{w}_1$ for all $k\geqslant \tau$.
\end{definition}

\begin{lemma}\label{lem:eventoglob}
	Let $\B$ be an LTI system described by \eqref{eqn:behaviour}. Then, for $N\in \N$, $\bs{w}_2$ is eventually locally $N$-reconstructible from $\bs{w}_1$ if and only if $\bs{w}_2$ is globally reconstructible from $\bs{w}_1$.
\end{lemma}
\begin{proof}
	(If) Assume that $\bs{w}_2$ is globally reconstructible from $\bs{w}_1$. 
	By \Cref{prop: recons_monomic_DBOform}, $\exists~N\in \N$ such that for $P(\sigma)=\sum_{i=0}^{d}P_i\sigma^i,~ P_i\in \R^{q_2\times q_1}$, the system $\begin{bsmallmatrix}
		-P(\sigma)&~\sigma^NI_{q_2}
	\end{bsmallmatrix}\begin{bsmallmatrix}
		\bs{w}_1\\\widehat{\bs{w}}_2
	\end{bsmallmatrix}=0$ forms a consistent DBO for $\B$.
	Thus, every $\bs{w}\in \B$ satisfies $\bs{w}_2(N)=\sum_{i=0}^{d}P_i\bs{w}_1(i)$. Hence, for $\tau=\max\{d,N\}$, $\bs{w}_2$ is $k$-locally $N$-reconstructible from $\bs{w}_1$ for all $k\geqslant \tau$.
	
	(Only if)
	 By \Cref{lem:local_DBO}, eventual local $N$-reconstructibility implies existence of a local DBO for some $k\in \N$ with $R_{\rm LDBO}\coloneqq \begin{bsmallmatrix}
		r_0&r_1&\cdots&r_k
	\end{bsmallmatrix}, ~r_i\in \R^{q_2\times q_1}$ for $i\in \bs{k}$. Then, due to time-invariance, every $\bs{w}\in \B$ satisfies $\bs{w}_2(t+N)=R_{\rm LDBO}\mathrm{col}(\bs{w}_1(t),\bs{w}_1(t+1),\ldots,\bs{w}_1(t+k))$ for all $t\in \N$.
 Define $R_{\rm DBO}(\sigma)=\sum_{i=0}^{k}r_i\sigma^i\in \R^{q_2\times q_1}[\sigma]$.
 Then, $\sigma^N\widehat{\bs{w}}_2=R_{\rm DBO}(\sigma)\bs{w}_1$ forms DBO for $\B$. Hence, by \Cref{prop: recons_monomic_DBOform}, $\bs{w}_2$ is globally reconstructible from $\bs{w}_1$.
\end{proof}

\Cref{lem:eventoglob} allows us to confirm global reconstructibility from $k$-local $N$-reconstructibility for some $k\in \N$, $N\in \bs{k}$.
 However, since an appropriate $k\in \N$ is not known, the failure of $k$-local $N$-reconstructibility for some $k$ does not imply failure of global reconstructibility.
In \Cref{lem:kmax} we provide, for a given $N\in \N$, a threshold on $k$ such that failure of $k$-local $N$-reconstructibility would imply failure of global $N$-reconstructibility. Then, in \Cref{thm:recons}, combining \Cref{lem:Nmax,lem:kmax}, we provide a general threshold on $k$.  

 Before presenting these results, we introduce some preliminaries. We first describe a \emph{driving-variable representation} (a DV representation) of $\B$\cite[Definition 3.1]{valcher2003note}. For $m,n, q_1, q_2\in \N$, let $\bs{x}\in (\R^n)^{\N}$, $\bs{u}\in (\R^m)^\N$ and $\bs{w}_i(t)\in \R^{q_i},~ i=1,2$ denote state, driving input and output vectors, respectively. The following representation
 \begin{align}\label{eq:driving_var}
 	\begin{split}
 		\bs{x}(t+1)&=A\bs{x}(t)+B\bs{u}(t)\\
 		\bs{w}_i(t)&=C_i\bs{x}(t)+D_i\bs{u}(t), \mbox{ for } i=1,2
 	\end{split}
 \end{align}
 with $A, B, C_1, C_2, D_1, D_2$ real matrices of suitable dimensions, is a DV representation of $\B$ if $\B=\left\{(\bs{w}_1, \bs{w}_2)\in (\R^{q_1+q_2})^\N \mid \exists~ \bs{x}\in (\R^n)^{\N},~ \bs{u}\in (\R^m)^\N \mbox{ s.t. }\right.$ $\left. (\bs{w}_1, \bs{w}_2, \bs{x}, \bs{u}) \mbox{ satisfies } \eqref{eq:driving_var}\right\}$. A minimal DV representation always exists with $m$ equal to the input cardinality and $n$ the order of $\B$ \cite[Section 3]{valcher2003note}; here, we use such a representation.
 
 Next, we define the \textit{weakly unobservable subspace} $\mathcal{V}_1\subseteq \R^n$ as follows\cite[Definition 7.8]{trentlemanbook}: $\mathcal{V}_1=\big\{\bs{x}(0)\in \R^n  \bigm| ~\exists~\bs{u}\in (\R^m)^\N\mbox{ s.t. } \bs{w}_1\equiv0 \mbox{ as per \eqref{eq:driving_var}} \big\}$.
In \Cref{prop:weaksub}, we give a matrix-based description of $\mathcal{V}_1$ used in the proof of \Cref{lem:kmax}. The proof is omitted as it follows from \cite[Lemma 1]{NaveenDDP}. Define for $k\in \N$, $\mathbb{M}_{1,k}\in \R^{q_1(k+1)\times (n+(k+1)m)}$ in the proposition as 
\begin{align}\label{eq:M-matrix}
	\mathbb{M}_{1,k}\!&=\!\scalebox{0.8}{$\left[\begin{array}{@{}l l l l l l @{}}
			C_1 & D_1 & 0 & 0 & \cdots & 0\\
			C_1A&C_1B&D_1&0&\cdots&0\\
			C_1A^2&C_1AB&C_1B&D_1&\cdots&0\\
			\vdots&\vdots&\vdots&\vdots&\ddots&\vdots\\
			C_1A^k&C_1A^{k-1}B&C_1A^{k-2}B&C_1A^{k-3}B&\cdots&D_1\\
		\end{array}\right]$}
\end{align}
 \begin{proposition}\label{prop:weaksub}
 	Consider an LTI system described by \eqref{eq:driving_var}. Then, the subspace $\mathcal{V}_1$ defined above is given by 
$ \mathcal{V}_1=\big\{\bs{x}(0)\in \R^n  \bigm| ~\exists~\bs{u}(0),\bs{u}(1),\ldots, \bs{u}(n-1)\mbox{ such that }$ $\mathbb{M}_{1,n-1}\mathrm{col}(
 			\bs{x}(0),\bs{u}(0),\bs{u}(1),\ldots,\bs{u}(n-1))=0\big\}.$
 \end{proposition}

\begin{lemma}\label{lem:kmax}
	Let $\B$ be an LTI system described by \eqref{eqn:behaviour} with order $n$. Let $N\in \N$. If $\bs{w}_2$ is globally $N$-reconstructible from $\bs{w}_1$ then $\bs{w}_2$ is $(N+n)$-locally $N$-reconstructible from $\bs{w}_1$.
\end{lemma}
\begin{proof} Assume that $\bs{w}_2$ is not $(N+n)$-locally $N$-reconstructible from $\bs{w}_1$. By \Cref{def:loc-recons}, $\exists~\bs{w}\in \B$ (i.e., $\bs{w}|_{N+n}\in \B_{N+n}$) such that $\bs{w}_1|_{N+n}\equiv 0$ but $\bs{w}_2(N)\neq 0$. By \Cref{def:recons}, it is enough to construct $\widetilde{\bs{w}}=(\widetilde{\bs{w}}_1,\widetilde{\bs{w}}_2)\in \B$ satisfying $\widetilde{\bs{w}}_1\equiv0$ while $\widetilde{\bs{w}}_2(t)\neq 0$ for some $t\geqslant N$.
 Using $\mathbb{M}_{1,k}$ in \eqref{eq:M-matrix}, and the DV representation in \eqref{eq:driving_var}, we have,
\begin{align}
		&\scalebox{0.9}{$\bs{w}_1|_{N+n}\!=\mathbb{M}_{1,N+n}\mathrm{col}(
		\bs{x}(0),\bs{u}(0),\bs{u}(1),\ldots,\bs{u}(N+n))\!=\!0$}\label{eq:notloc1}\\
	&	\bs{w}_2(N)=\begin{bsmallmatrix}
			C_2A^{N}&C_2A^{N-1}B&\cdots&C_2B&D_2
		\end{bsmallmatrix}\begin{bsmallmatrix}
		\bs{x}(0)\\\bs{u}(0)\\\vdots\\\bs{u}(N)
	\end{bsmallmatrix}\neq 0\label{eq:notloc2}
\end{align}
Note that $\bs{x}(N+1)=A^{N+1}\bs{x}(0)+A^{N}B\bs{u}(0)+\ldots+B\bs{u}(N)$. By the recursive nature of $\mathbb{M}_{1,k}$ and \eqref{eq:notloc1}, we have
 	$\mathbb{M}_{1,n-1}\mathrm{col}\left(\bs{x}(N+1), \bs{u}(N+1), \bs{u}(N+2),\ldots, \bs{u}(N+n) \right)$
 		$=\mathrm{col}(\bs{w}_1(N+1),\bs{w}_1(N+2),\ldots, \bs{w}_1(N+n))=0.$
 By \Cref{prop:weaksub}, it follows that $\bs{x}(N+1)\in \mathcal{V}_1$. Consequently, by the definition of $\mathcal{V}_1$, there exists input sequence $\widetilde{\bs{u}}\in (\R^m)^\N$ such that $\bs{w}_1\equiv 0$ as per \eqref{eq:driving_var} starting from initial condition $\bs{x}(N+1)$. 
 Define a trajectory $\widetilde{\bs{w}}=(\widetilde{\bs{w}}_1, \widetilde{\bs{w}}_2)\in \B$ via \eqref{eq:driving_var} with initial state $\bs{x}(0)$ and the driving input $\{\bs{u}(0), \bs{u}(1), \ldots,\bs{u}(N), \widetilde{\bs{u}}(0),\widetilde{\bs{u}}(1), \widetilde{\bs{u}}(2) \ldots\}$. Then, it follows that $\widetilde{\bs{w}}_1\equiv 0$, while $\widetilde{\bs{w}}_2(N)\neq 0$ by \eqref{eq:notloc2}. 
\end{proof}

\Cref{lem:kmax} states that if $k\geq n$ and $N\in [0,k-n]$, $\bs{w}_2$ is not $k$-locally $N$-reconstructible from $\bs{w}_1$, then $\bs{w}_2$ is not globally $N$-reconstructible from $\bs{w}_1$. By \Cref{lem:Nmax}, we have $N\leqslant n$. Combining it with \Cref{lem:kmax}, we obtain the following result.

\begin{theorem}\label{thm:recons}
	Let $\B$ be an LTI system given by \eqref{eqn:behaviour} of order $n$.  Then, $\bs{w}_2$ is globally reconstructible from $\bs{w}_1$ if and only if for $N\in \boldsymbol{n}$, $\bs{w}_2$ is $2n$-locally $N$-reconstructible from $\bs{w}_1$.	
\end{theorem}
\begin{proof}
	(If) If $\bs{w}_2$ is $2n$-locally $N$-reconstructible from $\bs{w}_1$ then, by \Cref{lem:local_DBO}, the existence of local DBO, it follows that $\bs{w}_2$ is eventually locally $N$-reconstructible from $\bs{w}_1$. Hence, by \Cref{lem:eventoglob}, global reconstructibility follows. 
	
	(Only if) By \Cref{lem:Nmax}, global reconstructibility implies global $N$-reconstructibility for some $N\in \boldsymbol{n}$. Furthermore, by \Cref{lem:kmax}, global $N$-reconstructibility implies $(N+n)$-local $N$-reconstructibility. Since $N+n\leqslant 2n~\forall ~N\in \bs{n},$ $\bs{w}_2|_{2n}=0\implies \bs{w}_2|_{N+n}=0$, thus by \Cref{def:loc-recons}, for $N\in \boldsymbol{n}$, $\bs{w}_2$ is $2n$-locally $N$-reconstructible from $\bs{w}_1$.	
\end{proof}

\begin{remark}\label{rem:worstcase}
	The threshold $2n$ in \Cref{thm:recons} often turns out to be a weak one. For one, the depth of the matrix $\mathbb{M}_{1,d}$ in \Cref{prop:weaksub} can often be taken to be less than $n-1$. For example, for full row rank $D_1$, $d=0$ suffices. If $d^*$ is the required depth of $\mathbb{M}_{1,d}$ in \Cref{prop:weaksub}, \Cref{lem:kmax} requires $(N+d^*+1)$-local $N$-reconstructibility. Furthermore, the DV representation approach taken in \Cref{lem:kmax} may turn out to be suboptimal as far as the threshold on $k$ is concerned. In our extensive simulation studies, we have always found a threshold less than $2n$. However, a concrete result still eludes us.
\end{remark}

\section{Data-driven test for reconstructibility}\label{sec:datarecons}

We consider an LTI system $\B$ described by \eqref{eqn:behaviour} where the matrices $R_1(\sigma)$ and $R_2(\sigma)$ are unknown. We assume that we have pre-collected historical data of both the relevant and measured variables\footnote{The relevant variables may be unavailable in standard working conditions due, for example, to the high cost of sensors or security reasons, but it is possible to collect their finite data offline (see, for example, \cite{DBO_10720866, DBO_turan2021data, DBO_shi2022data} ).}. We further assume that the order of the system $n$ is known; while restrictive, this assumption enables a data-driven solution to the DBO problem and is consistent with related work on higher-order difference equations (see, e.g. \cite{van2023behavioral}). 
In this section, using a historical data, we present a necessary and sufficient condition for global reconstructibility. 

Suppose that finite-length historical data from $K$ different trajectories, possibly of different lengths, are available. Let $T_0,T_1,\ldots,T_K$ be non-negative integers, not necessarily distinct. For each $i\in \bs{K}$, let $\widetilde{\bs{w}}^i|_{T_i}$ denote a $q$-variate time-series defined on the interval $[0, T_i]$, i.e., $\widetilde{\bs{w}}^i|_{T_i}\in \B_{T_i}$. The overall finite sample multiple-trajectory historical dataset\footnote{ Unlike most existing works, we do not assume $\B$ to be controllable; we therefore require a multiple-trajectory dataset so as to ensure a suitable ``informativity'' condition is satisfied by it; \textcolor{black}{see \cite{markovsky2022identifiability} for more details.}} is given by
$
\mathcal{W}\coloneqq \{\widetilde{\bs{w}}^0|_{T_0},\ldots,\widetilde{\bs{w}}^K|_{T_K}\}.	
$
The partition of $\mathcal{W}$ corresponding to $\bs{w}=(\bs{w}_1,\bs{w}_2)$ is 
$
{\mathcal{W}}_1\coloneqq \{\widetilde{\bs{w}}_1^0|_{T_0},\ldots,\widetilde{\bs{w}}_1^K|_{T_K}\},
$ and $
{\mathcal{W}}_2\coloneqq \{\widetilde{\bs{w}}_2^0|_{T_0},\ldots,\widetilde{\bs{w}}_2^K|_{T_K}\}
$.     
Fix $k\in\N$ such that $k\leqslant \min\{T_0, T_1,\ldots, T_K\}$. For each trajectory data $\widetilde{\bs{w}}^i|_{T_i}$, we define the Hankel matrix of depth $k$, denoted by $\hankel{{\bs{w}}^i|_{T_i}}{k}\in \R^{q(k+1)\times (T_i-k+1)},$ as follows:
\begin{align*}
	\hankel{\widetilde{\bs{w}}^i|_{T_i}}{k}
	=\scalebox{0.8}{$\left[
		\begin{array}{@{} l l l l @{}}
			{\widetilde{\bs{w}}}^i(0) & {\widetilde{\bs{w}}}^i(1) & \cdots & {\widetilde{\bs{w}}}^i(T_i - k) \\
			{\widetilde{\bs{w}}}^i(1) & {\widetilde{\bs{w}}}^i(2) & \cdots & {\widetilde{\bs{w}}}^i(T_i - k + 1) \\
			\vdots & \vdots & \ddots & \vdots \\
			{\widetilde{\bs{w}}}^i(k) & {\widetilde{\bs{w}}}^i(k+1) & \cdots & {\widetilde{\bs{w}}}^i(T_i)
		\end{array}
		\right]$}.
\end{align*}
For $N\in \N$, let the $N$-th block row of $\hankel{\widetilde{\bs{w}}^i|_{T_i}}{k}$ is given by
$ (\hankel{\widetilde{\bs{w}}^i|_{T_i}}{k})_N=\begin{bsmallmatrix}
		{\widetilde{\bs{w}}}^i(N) & {\widetilde{\bs{w}}}^i(N+1) & \cdots & {\widetilde{\bs{w}}}^i(T_i-k+N)
\end{bsmallmatrix}$.
Using the multiple-trajectory dataset ${\mathcal{W}}$, define the {\em mosaic-Hankel} matrix of depth $k$, denoted by $\hankel{{\mathcal{W}}}{k} \in \R^{q(k+1)\times \sum_{i\in \bs{K}}(T_i-k+1)}$, as:
$
	\hankel{{\mathcal{W}}}{k}\coloneqq \begin{bsmallmatrix}
		\hankel{\widetilde{\bs{w}}^0|_{T_0}}{k}&\hankel{\widetilde{\bs{w}}^1|_{T_1}}{k}&~\cdots~&\hankel{\widetilde{\bs{w}}^K|_{T_K}}{k}
	\end{bsmallmatrix}.
$
For $N\in \N$, $(\hankel{{\mathcal{W}}}{k})_N$ denotes $N$-th block row of $\hankel{{\mathcal{W}}}{k}$ with entries $(\hankel{\widetilde{\bs{w}}^i|_{T_i}}{k})_N$, $i\in \bs{K}$. The matrices $\hankel{{\mathcal{W}}_1}{k}$ and $(\hankel{{\mathcal{W}}_2}{k})_N$ are defined analogously.
Let $n$ be the order of $\B$.
 To apply \Cref{thm:recons}, we need $\mathcal{W}$ that represents the subspace
$\mathbb{B}\!=\left\{\begin{bsmallmatrix}
	\bs{w}_1|_{2n}\\
	\bs{w}_2|_n
\end{bsmallmatrix} \mid (\bs{w}_1,\bs{w}_2)\in \B \right\}$. Let $(\hankel{\mathcal{W}_2}{2n})_{0:n}=\mathrm{col}\left((\hankel{\mathcal{W}_2}{2n})_{0},\ldots, (\hankel{\mathcal{W}_2}{2n})_{n}\right)$.  Note that every dataset satisfies $\mathrm{Im}\begin{bsmallmatrix}
	\hankel{{\mathcal{W}}_1}{2n}\\
	(\hankel{\mathcal{W}_2}{2n})_{0:n}
\end{bsmallmatrix} \!\subseteq\! \mathbb{B}$; we impose the following.
\begin{assumption}\label{assump:1}
	Let $\B$ be an LTI system and let ${\mathcal{W}}$ be its historical trajectory dataset. Let $\mathbb{B}$ be the subspace defined above. We assume that the dataset satisfies $
	\mathrm{Im}\begin{bsmallmatrix}
		\hankel{{\mathcal{W}}_1}{2n}\\
		(\hankel{\mathcal{W}_2}{2n})_{0:n}
	\end{bsmallmatrix}=\mathbb{B}.
	$
\end{assumption}
\begin{remark}
\textcolor{black}{Note that \Cref{assump:1} is presented in a way that is suitable for the data-driven test for reconstructibility (see \Cref{thm:datarecons}). At a more abstract level, \Cref{assump:1} is just another form of the well-known idea of informativity (see \cite{van2020data}), albeit tailored to the case of reconstructibility. Discussion on some important aspects of this assumption, e.g., practicality, verification is warranted at this stage. We provide two sufficient conditions, each of which independently ensures \Cref{assump:1}. To verify these conditions, we assume either direct access to the input, or the knowledge of the order $n$ and the input cardinality $m$ of the system is available.}
\begin{itemize}
	\item \textcolor{black}{For a controllable system with input $\bs{u}$, use a persistently exciting input of order $3n$, i.e., $\hankel{\bs{u}}{3n-1}=3n$ \cite{willems2005note}}.
	\item \textcolor{black}{	Use the dataset satisfying $\mathrm{rank}(\hankel{{\mathcal{W}}}{2n})=m(2n+1)+n$ so that $\mathrm{Im}\hankel{{\mathcal{W}}}{2n}=\B_{2n}$\cite[Corollary 21]{markovsky2022identifiability}. }
\end{itemize}

\end{remark}

\begin{theorem}\label{thm:datarecons}
	Let $\B$ be an LTI system with order $n$ with a dataset ${\mathcal{W}}$ satisfying \Cref{assump:1}. Then, $\bs{w}_2$ is globally reconstructible from $\bs{w}_1$ in $\B$ if and only if for some $N\in \boldsymbol{n}$
	\begin{align}\label{eq:datarecons}
		\mathrm{rowspan}((\hankel{{\mathcal{W}_2}}{2n})_N)\subseteq\mathrm{rowspan}(\hankel{{\mathcal{W}}_1}{2n}).
	\end{align}
\end{theorem}
\begin{proof}
Since $\mathcal{W}$ satisfies \Cref{assump:1}, for every $\bs{w}\in \B$ there exists $v\in \R^{\sum_{i\in \bs{K}}(T_i-2n+1)}$ such that, for all $N\in \boldsymbol{n}$,	$
		\begin{bsmallmatrix}
			\bs{w}_1|_{2n}\\
			\bs{w}_2(N)
		\end{bsmallmatrix}=\begin{bsmallmatrix}
			\hankel{{\mathcal{W}}_1}{2n}\\
			(\hankel{{\mathcal{W}}}{k})_N
		\end{bsmallmatrix}v.$
	Thus, for $N\in \boldsymbol{n}$, $\left(\bs{w}_1|_{2n}=0\implies\right.$ $\left. \bs{w}_2(N)=0\right)$ if and only if $\ker(\hankel{{\mathcal{W}}_1}{2n})\subseteq \ker((\hankel{{\mathcal{W}_2}}{k})_N) $. Hence, by \Cref{def:loc-recons}, for $N\in \bs{n}$, $\bs{w}_2$ is $2n$-locally $N$-reconstructible if and only if \eqref{eq:datarecons} is satisfied. The result then follows from \Cref{thm:recons}.   
\end{proof}

\textcolor{black}{Dualising the rowspan condition in \Cref{thm:datarecons} directly yields the local DBO in \Cref{lem:local_DBO}, from which a global DBO can be obtained (see Algorithm 1).}
\begin{figure*}[t]
	\centering
	\subfloat[$\bs{x}_1$ estimate using $\Sigma_1$]{%
		\includegraphics[width=0.235\textwidth]{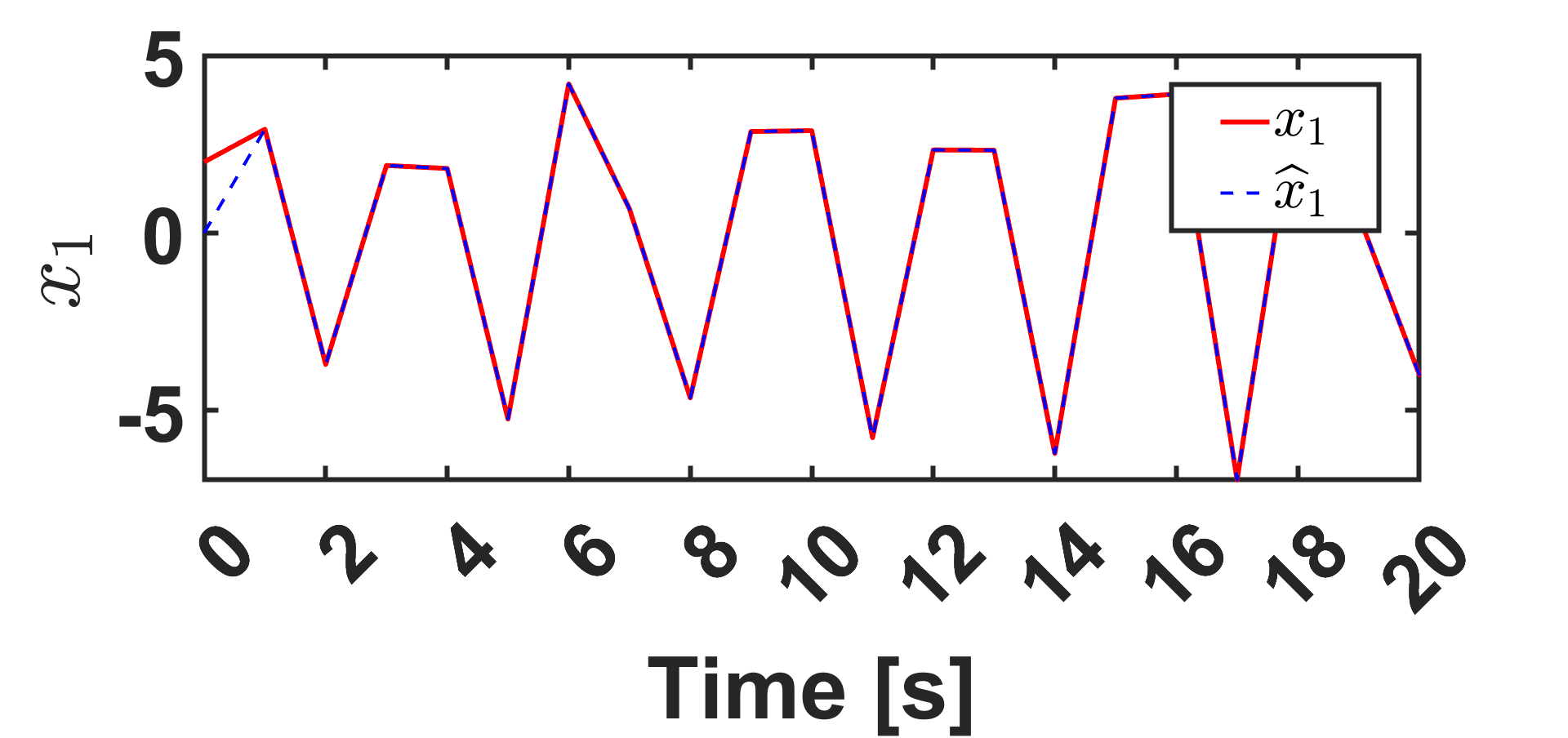}}
	\hfill
	\subfloat[$\bs{x}_2$ estimate using $\Sigma_1$]{%
		\includegraphics[width=0.235\textwidth]{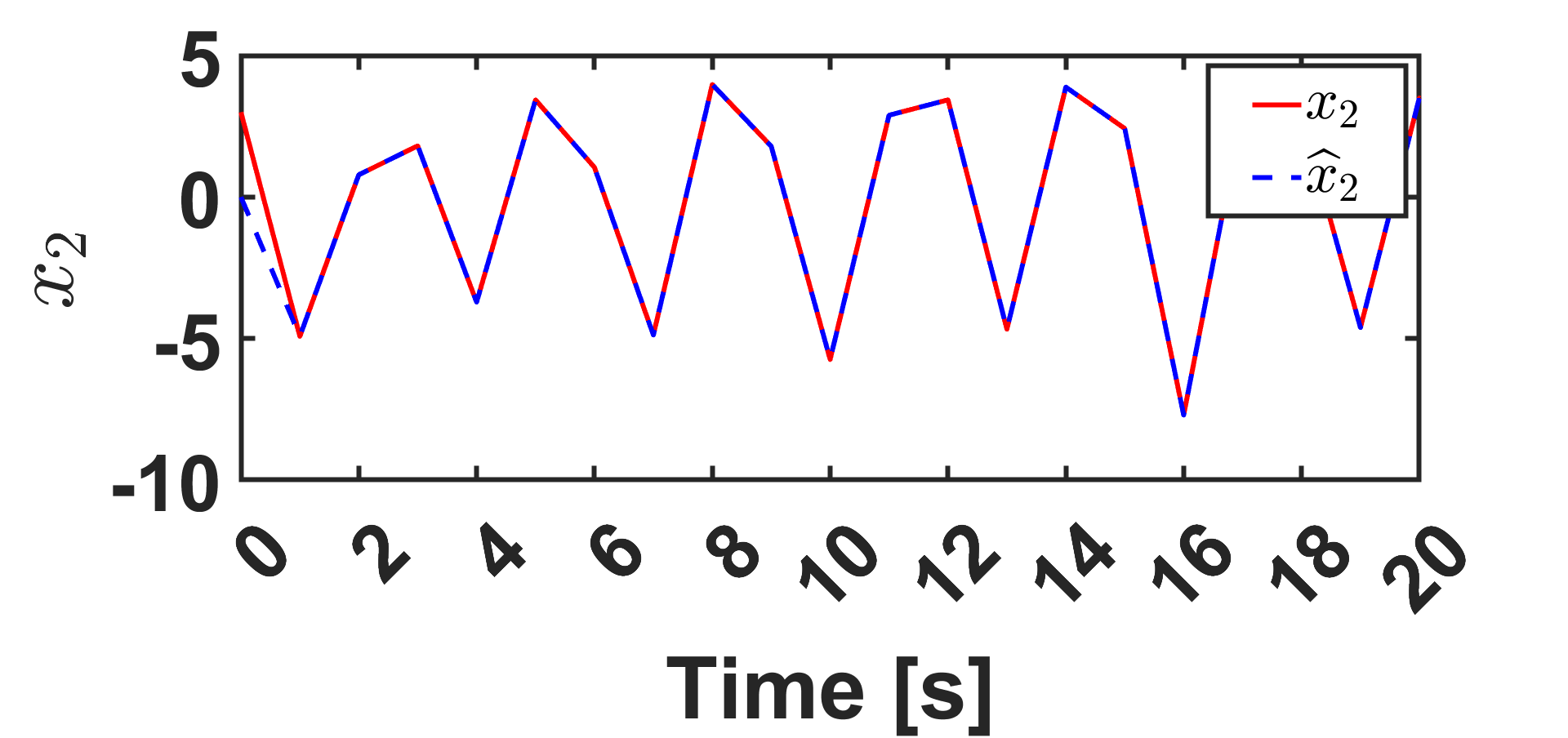}}
	\hfill
	\subfloat[$\bs{x}_1$ estimate using $\Sigma_2$]{%
		\includegraphics[width=0.235\textwidth]{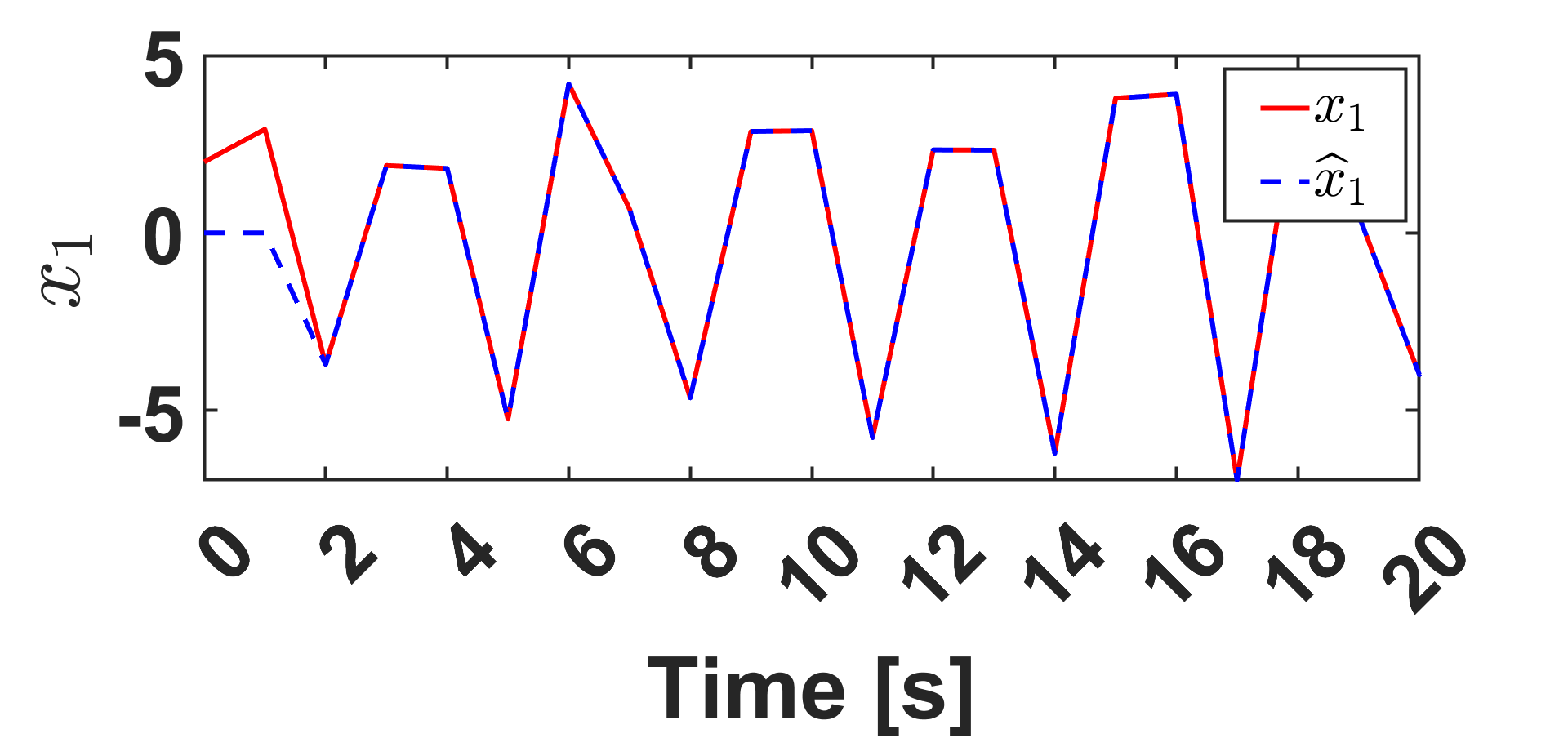}}
	\hfill
	\subfloat[$\bs{x}_2$ estimate using $\Sigma_2$]{%
		\includegraphics[width=0.235\textwidth]{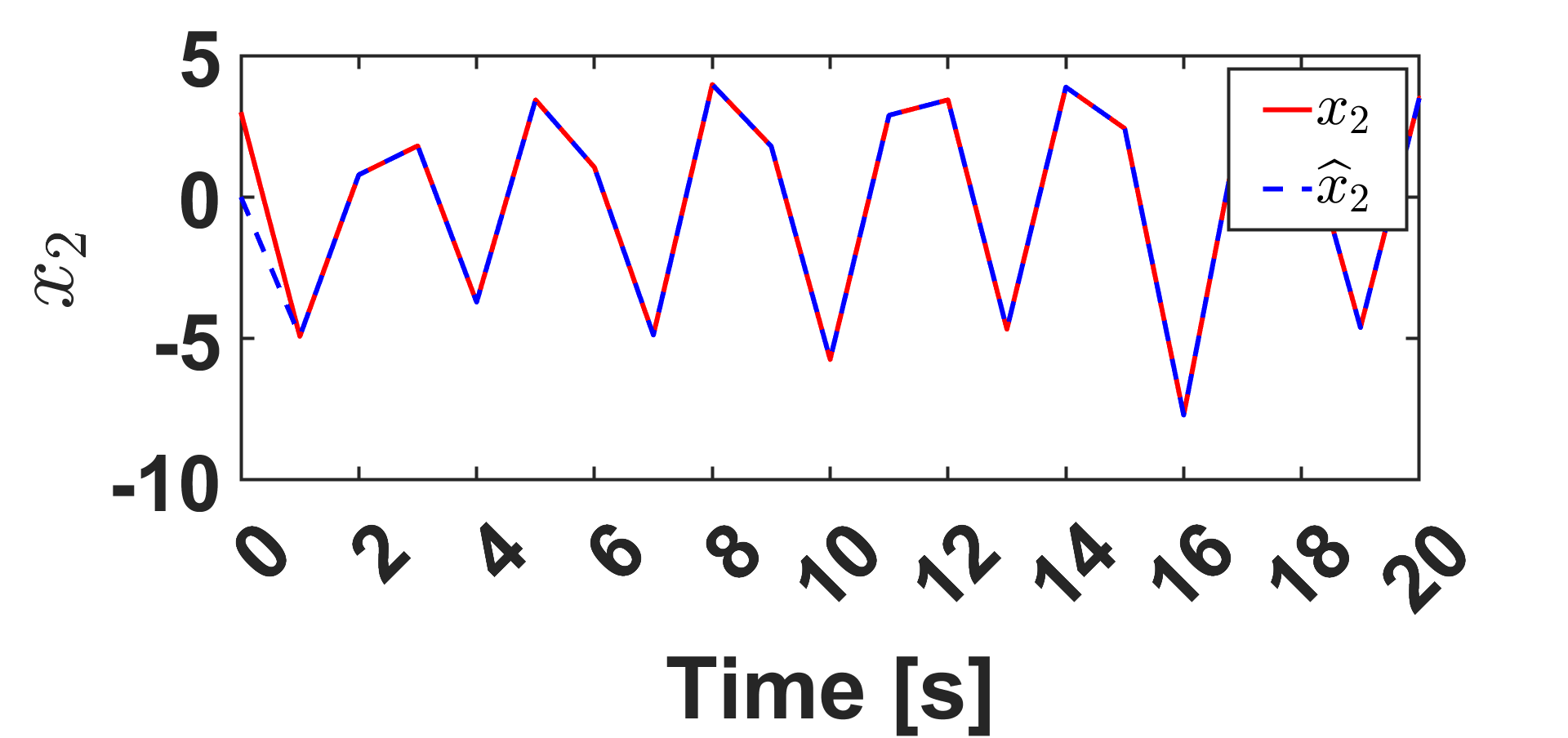}}
	
	\caption{State estimation using the two DBOs in \Cref{sec:example}}
	\label{fig:eg}
\end{figure*}
\begin{remark}
\textcolor{black}{By \Cref{rem:worstcase}, $2n$ is a worst-case threshold for \Cref{thm:recons}, leading to the following observation on the under- and overestimation of $n$ by $\widehat{n}$ for $\mathcal{W}$ satisfying \Cref{assump:1}. Consider the following set inclusion}
\begin{align}\label{eqn:n_dep}
	\mathrm{rowspan}((\hankel{{\mathcal{W}_2}}{2\widehat{n}})_N)\!\subseteq\!\mathrm{rowspan}(\hankel{{\mathcal{W}}_1}{2\widehat{n}})
\end{align}
\textcolor{black}{If $\widehat{n}$ is an underestimation, i.e., $\widehat{n}<n$, then the satisfaction of \eqref{eqn:n_dep} for some $N\in \bs{\widehat{n}}$ implies global reconstructibility; otherwise, the test is inconclusive. On the other hand, if $\widehat{n}$ is an overestimation, i.e., $\widehat{n}\geqslant n$, the satisfaction of \eqref{eqn:n_dep} for some $N\in \bs{\widehat{n}}$ implies global reconstructibility otherwise we conclude failure of global reconstructibility.}
\end{remark}
\section{Deadbeat observer design from finite data}\label{sec:algorithm}
In this section, using a dataset satisfying \eqref{eq:datarecons}, we construct a DBO: $\sigma^N\widehat{\bs{w}}_2=P(\sigma)\bs{w}_1$ where $N\in \N$ and $P(\sigma) \in\R^{q_2\times q_1}[\sigma]$. 
 We then prove in \Cref{lemma:onetoallDBO} that all other DBOs can be obtained from a single DBO, yielding a data-based parametrisation of all DBOs. First, we state a model-based parametrisation from \cite[Theorem 7]{valcher_bisiacco2006behavioral}.
Let $\B$ be given by \eqref{eqn:behaviour}, with $R_2(\sigma)$ right monomic. Then, $\exists~U(\sigma)\in \R^{g\times g}[\sigma]$ unimodular such that $U(\sigma)R(\sigma)=\begin{bsmallmatrix}
	L_1(\sigma)&\Delta(\sigma)\\
	L_0(\sigma)&0
\end{bsmallmatrix}\eqqcolon\overline{R}(\sigma)$ and $\B=\ker \overline{R}(\sigma)$
where $\Delta(\sigma)\in \R^{q_2\times q_2}[\sigma]$ is square monomic.
\begin{proposition}[\cite{valcher_bisiacco2006behavioral}]\label{prop:parametDBO}
	Let $\B=\ker \overline{R}(\sigma)$ as above. For $P(\sigma)\!\in \!\R^{q_2\times q_1}[\sigma]$, $Q(\sigma)\!\in \! \R^{q_2\times q_2}[\sigma]$,  $	Q(\sigma)\widehat{\bs{w}}_2=P(\sigma)\bs{w}_1$ is a DBO for $\B$ if and only if there exist Laurent polynomial matrices $X(\sigb)\in \R^{q_2\times q_2}[\sigb]$ and $Y(\sigb)\in \R^{q_2\times (g-q_2)}[\sigb]$ such that $\begin{bsmallmatrix}
		-P(\sigma)&~Q(\sigma)
	\end{bsmallmatrix}=\begin{bsmallmatrix}
		X(\sigb)&~Y(\sigb)
	\end{bsmallmatrix}\overline{R}(\sigma)$ where $Q(\sigma)=X(\sigb)\Delta(\sigma)$ is a monomic.
\end{proposition}
\begin{lemma}\label{lemma:onetoallDBO}
Let $\B=\ker \overline{R}(\sigma)$ as above.
	Let $Q(\sigma)\widehat{\bs{w}}_2=P(\sigma)\bs{w}_1$  be a DBO for $\B$. Then,  $\widetilde{Q}(\sigma)\widehat{\bs{w}}_2=\widetilde{P}(\sigma)\bs{w}_1$  is a DBO for $\B$ if and only if there exist Laurent polynomial matrices $\widehat{X}(\sigb)\in \R^{q_2\times q_2}[\sigb]$ and  $\widehat{Y}(\sigb)\in \R^{q_2\times (g-q_2) }[\sigb]$ such that $\widetilde{Q}(\sigma)=\widehat{X}(\sigb)Q(\sigma)$ is square monomic and
	\begin{align}\label{eqn:1DBOtoall}
		\begin{bsmallmatrix}
			-\widetilde{P}(\sigma)&~\widetilde{Q}(\sigma)
		\end{bsmallmatrix}=\begin{bsmallmatrix}
			\widehat{X}(\sigb)&~\widehat{Y}(\sigb)
		\end{bsmallmatrix}\begin{bsmallmatrix}
			-P(\sigma)&~	Q(\sigma)\\
			L_0(\sigma)&0
		\end{bsmallmatrix}.
	\end{align}	
\end{lemma}

\begin{proof}
	By \Cref{prop:parametDBO},
	  $\exists ~X(\sigb),$ $Y(\sigb)$ such that
	\begin{align}\label{eq:lemdbo1}
		\begin{split}
			\begin{bsmallmatrix}
				-P(\sigma)&~Q(\sigma)
			\end{bsmallmatrix}=\begin{bsmallmatrix}
				X(\sigb)&~Y(\sigb)
			\end{bsmallmatrix}\overline{R}(\sigma).
		\end{split}
	\end{align}
	(If) Assume $\widetilde{Q}(\sigma)$ is monomic and \eqref{eqn:1DBOtoall} holds. By \eqref{eq:lemdbo1}, we have
	$\begin{bsmallmatrix}
		-\widetilde{P}(\sigma)&~\widetilde{Q}(\sigma)
	\end{bsmallmatrix}= \begin{bsmallmatrix}
	\widehat{X}(\sigb)&~\widehat{Y}(\sigb)
\end{bsmallmatrix}\begin{bsmallmatrix}
	X(\sigb)&~Y(\sigb)\\
	0& I_{g-q_2}
\end{bsmallmatrix}\overline{R}(\sigma).$
Thus, by \Cref{prop:parametDBO}, the claim follows.\\
	(Only if ) Assume  $
	\widetilde{Q}(\sigma)\widehat{\bs{w}}=\widetilde{P}(\sigma)\bs{w}_1
	$ is a DBO for $\B$. Then, by \Cref{prop:parametDBO}, $\exists ~\widetilde{X}(\sigb)$, $\widetilde{Y}(\sigb)$ of suitable sizes such that $	\begin{bsmallmatrix}
		\widetilde{P}(\sigma)&~\widetilde{Q}(\sigma)
	\end{bsmallmatrix}=\begin{bsmallmatrix}
		\widetilde{X}(\sigb)&~\widetilde{Y}(\sigb)
	\end{bsmallmatrix}\overline{R}(\sigma).$
In \eqref{eq:lemdbo1}, since $Q(\sigma)=X(\sigb)\Delta(\sigma)$ and $\Delta(\sigma)$ are monomic, $X(\sigb)$ is invertible over $\R[\sigb]$.
Thus, we have $\begin{bsmallmatrix}
		-\widetilde{P}(\sigma)&~\widetilde{Q}(\sigma)
\end{bsmallmatrix}= \begin{bsmallmatrix}
\widetilde{X}(\sigma)&~\widetilde{Y}(\sigma)
\end{bsmallmatrix}\begin{bsmallmatrix}
	X(\sigb)&~Y(\sigb)\\
	0&I_{g-q_2}
\end{bsmallmatrix}^{-1} \begin{bsmallmatrix}
-P(\sigma)&~	Q(\sigma)\\
L_0(\sigma)&0
\end{bsmallmatrix}$. 
\end{proof}
\noindent\textbf{Algorithm 1:} \textit{Data-driven DBO construction}\\
\underline{Input}: Historic dataset $\mathcal{W}$ satisfying \eqref{eq:datarecons} in \Cref{thm:datarecons}\\
\underline{Output}: A consistent DBO and parametrisation of all DBOs\\
(A) \textit{Construction of a consistent DBO}
\begin{enumerate}
	\item Find $N^*\in \bs{n}$ and $R_{\rm LDBO}=\begin{bsmallmatrix}
		R_0&R_1&\cdots&R_{2n}
	\end{bsmallmatrix},~R_i\in \R^{q_2\times q_1}$   such that 
		$(\hankel{{\mathcal{W}_2}}{2n})_{N^*}=R_{\rm LDBO}\hankel{{\mathcal{W}}_1}{2n}$.
	\item Define $R_{\rm DBO}(\sigma)=\sum_{i=0}^{2n}R_i\sigma^i$. Then, the system
	$\sigma^{N^*}\widehat{\bs{w}}_2=R_{\rm  DBO}(\sigma)\bs{w}_1$ is a DBO for $\B$.
\end{enumerate}
(B) \textit{Parametrisation of all DBOs}
\begin{enumerate}
	\item Find $L=\begin{bsmallmatrix}
		\ell_0&\ell_1&\cdots&\ell_{n}
	\end{bsmallmatrix}\in \R^{\bullet\times q_1(n+1)},~\ell_i\in \R^{\bullet\times q_1}$ such that
	$\mathrm{rowspan} L=\mathrm{leftker}\hankel{{\mathcal{W}}_1}{n}$. 
	\item Define $L(\sigma)=\sum_{i=0}^{n}\ell_i\sigma^i$. Choose matrices $\widehat{X}(\sigb)$, $\widehat{Y}(\sigb)$ of suitable sizes such that ${Q}(\sigma)\coloneqq\sigma^{N^*}\widehat{X}(\sigb)$ is square monomic and
	$\begin{bsmallmatrix}
		-P(\sigma)&~Q(\sigma)
	\end{bsmallmatrix}\coloneqq\begin{bsmallmatrix}
		\widehat{X}(\sigb)&~\widehat{Y}(\sigb)
	\end{bsmallmatrix}\begin{bsmallmatrix}
	-R_{\rm DBO}(\sigma)&~~	\sigma^{N^*}I_{q_2}\\
	L(\sigma)& 0
\end{bsmallmatrix}$. Then, $Q(\sigma)\widehat{\bs{w}}_2=P(\sigma)\bs{w}_1$ is DBO for $\B$.
\end{enumerate}
\begin{theorem}
Let $\mathcal{W}$ be a historic dataset of an LTI system $\B$ satisfying \eqref{eq:datarecons}. Then, Algorithm 1 computes a consistent DBO and a parametrisation of all DBOs for $\B$.
\end{theorem}
\begin{proof}
	\emph{(A)} follows from \Cref{thm:datarecons} and \Cref{lem:eventoglob}. Note that $L_0(\sigma)$ in \Cref{lemma:onetoallDBO} generates all $g(\sigma)\in \R^{1\times q_1}[\sigma]$ such that $g(\sigma)\bs{w}_1=0~\forall~(\bs{w}_1,\bs{w}_2)\in \B$, thus\footnote{For $P(\sigma)=[p_{ij}(\sigma)]\in \R^{m\times n}[\sigma]$, the degree of $P(\sigma)$, denoted by $\deg P(\sigma)$, is defined as $\deg P(\sigma)\coloneqq \max_{ij}\deg p_{ij}(\sigma)$. } $\deg L_0(\sigma)\leqslant n$. Since $\mathcal{W}$ satisfies \Cref{assump:1}, $L(\sigma)$ computed in step $2)$ of \emph{(B)} represents $L_0(\sigma)$ in \Cref{lemma:onetoallDBO}\cite[Section XIV]{willems1991paradigms}. 
\end{proof}

\section{Example}\label{sec:example}
	In this section, we illustrate our results for state estimation in the presence of an unknown input (disturbance). We consider the state-space system with state $\bs{x}\in (\R^2)^{\N}$, unknown input (disturbance) $\bs{d}\in \R^{\N}$ and output $\bs{y}\in \R^{\N}$ given by   
$
			\bs{x}(k+1)=\begin{bsmallmatrix}
				0 & ~~1 \\
				-1 &~ -1
			\end{bsmallmatrix}\bs{x}+
		\begin{bsmallmatrix}
		-1 \\
		~1
	\end{bsmallmatrix}\bs{d},~
\bs{y}=\begin{bsmallmatrix} 1 & ~0 \end{bsmallmatrix}\bs{x}.
	$
	We are interested in a DBO of $\bs{x}$ from $\bs{y}$. Thus, we have  $\bs{w}_1=\bs{y}$, $\bs{w}_2=\bs{x}$ and $\bs{d}$ is irrelevant (see \Cref{rem:irrelevant}).
			By choosing $\bs{d}(t)$ to be zero-mean unit-variance Gaussian sequence
				 and initial condition $\bs{x}(0)=\begin{bsmallmatrix}
					2\\3
				\end{bsmallmatrix}$, for $T=11$ (\textcolor{black}{ensuring square/fat Hankel matrices}), we collect a single trajectory data ${\mathcal{W}}_1={\bs{y}|_T}$ and ${\mathcal{W}}_2={\bs{x}|_T}$ satisfying \Cref{assump:1}  (In this case, {\Cref{assump:1} holds generically for random disturbances; see \cite[Remark 6]{DBO_10720866}}).
				 We now apply Algorithm 1 to obtain DBO for $\B$ and its parametrisation \textcolor{black}{using MATLAB's direct commands ($\mathtt{null, linsolve, rank}$) under default machine-precision settings.}
				We obtain that for $N=1$, the Hankel matrices $(\hankel{{\mathcal{W}_2}}{4})_N\in\R^{2\times 7}$ with \textcolor{black}{$\mathrm{rank}(\hankel{{\mathcal{W}_2}}{4})=2$} and $\hankel{{\mathcal{W}}_1}{4}\in\R^{5\times 7}$ with \textcolor{black}{$\mathrm{rank}(\hankel{{\mathcal{W}}_1}{4})=5$} satisfy \eqref{eq:datarecons} and for
				$
				R_{\rm LDBO} =
				\begin{bsmallmatrix}
					0 & 1 &  0&0&0 \\
					-1     &  -1 &  0      &  0&0
				\end{bsmallmatrix}
				$, we have
				$(\hankel{{\mathcal{W}_2}}{4})_N =R_{\rm LDBO}\hankel{{\mathcal{W}}_1}{4}$.
				Thus, $R_{\rm DBO}(\sigma)=\begin{bsmallmatrix}
						\sigma\\
					-1-\sigma
				\end{bsmallmatrix}$ and $\Sigma_1:$
				 $\sigma\widehat{\bs{x}}=R_{\rm DBO}(\sigma)\bs{y}$ is a DBO for $\B$. \textcolor{black}{It can be theoretically verified that $\Sigma_1$ forms a DBO for $\B$ (see \cite{valcher_bisiacco2006behavioral} for more details). Also, we obtained the same DBO across several datasets generated from different initial conditions, which further upholds the correctness of the algorithm.}
				Since, $\mathcal{H}_n({\mathcal{W}_1})$ is full row rank, $L(\sigma)=0$. Thus, for any $X(\sigb)\in \R^{2\times 2}[\sigma]$ with $\sigma X(\sigb)$ monomic, define $\begin{bsmallmatrix}
					-P(\sigma)&~Q(\sigma)
				\end{bsmallmatrix}=X(\sigb)\begin{bsmallmatrix}
					-R_{\rm DBO}(\sigma)&~\sigma I_2
				\end{bsmallmatrix}$. Then,  $Q(\sigma)\widehat{\bs{x}}=P(\sigma)\bs{y}$ is a DBO for $\B$ and every DBO can be obtained this way. \textcolor{black}{ For $X(\sigma)=\begin{bsmallmatrix}
					\sigma&0\\
					0&1
			\end{bsmallmatrix}$, we obtain another DBO $\Sigma_2:$ $\begin{bsmallmatrix}
			\sigma^2 \widehat{\bs{x}}_1\\ \sigma \widehat{\bs{x}}_2
		\end{bsmallmatrix}=\begin{bsmallmatrix}
		\sigma^2\\-1-\sigma
	\end{bsmallmatrix}\bs{y}$.  \Cref{fig:eg} provides the resulting state estimates using these two DBOs, with $\widehat{\bs{x}}(0)=(0,0)$. }  
\begin{remark}\label{rem:noisy data}
\textcolor{black}{	Note that while the process noise can be dealt with as an irrelevant variable (see \Cref{rem:irrelevant}), in the presence of the measurement noise, the ranks of the Hankel matrices are inflated. Therefore, the condition in \Cref{thm:datarecons} does not work directly. However, the rowspan condition suggests that we can reformulate it into an error-in-variables problem and solve it using the structured low rank approximation while preserving the Hankel structure of the matrices \cite{slra-software} to obtain an approximate observer using a noisy dataset. An in-depth analysis of the noisy historical data will be taken up in future.}
\end{remark}
	
				\section{Conclusion}\label{sec:conclusion}
				We studied constructing DBOs for relevant variables from measured variables using historical data for higher-order LTI systems. By introducing the notion of `local reconstructibility', we established a necessary and sufficient condition for the existence of a DBO from finite data. Under a suitable assumption on the historical data, we provided an algorithm to obtain a DBO and a parametrisation of all DBOs. Finally, we illustrated the results using state estimation in the presence of unknown inputs.
			Future work includes data-driven tests for the existence and construction of lower complexity, non-consistent DBOs \cite{valcher_bisiacco2006behavioral} and causal DBOs for real-time implementation.
	
	\bibliography{ref.bib}
\bibliographystyle{ieeetr}

\end{document}